\documentclass[11pt]{article}
\usepackage{bm,latexsym,amssymb,float,amsmath,amsthm,enumerate,cite,geometry,algorithm2e,tikz}
\newtheorem{theorem}{Theorem}
\newtheorem{proposition}[theorem]{Proposition}
\newtheorem{lemma}[theorem]{Lemma}

\usepackage{setspace}
\allowdisplaybreaks

\usetikzlibrary{decorations.text}
\usetikzlibrary{decorations.pathmorphing}
\usepackage{pgfplots}
\usetikzlibrary{calc}

\begin{document}
\onehalfspace

\title{Induced Subforests and Superforests}
\author{Dieter Rautenbach \and Florian Werner}
\date{}

\maketitle
\vspace{-10mm}
\begin{center}
Institute of Optimization and Operations Research, Ulm University, Ulm, Germany\\
\texttt{$\{$dieter.rautenbach,florian.werner$\}$@uni-ulm.de}
\end{center}

\begin{abstract}
Graph isomorphism, 
subgraph isomorphism, and
maximum common subgraphs
are classical well-investigated objects.
Their (parameterized) complexity and efficiently tractable cases have been studied.
In the present paper, for a given set of forests,
we study maximum common induced subforests and 
minimum common induced superforests.
We show that finding a maximum subforest is NP-hard already for two subdivided stars
while finding a minimum superforest is tractable for two trees
but NP-hard for three trees.
For a given set of $k$ trees, we present an efficient greedy
$\left(\frac{k}{2}-\frac{1}{2}+\frac{1}{k}\right)$-approximation algorithm 
for the minimum superforest problem.
Finally, we present a polynomial time approximation scheme 
for the maximum subforest problem for any given set of forests.\\[3mm]
\end{abstract}

\section{Introduction}\label{sec1}

We consider finite, simple, and undirected graphs and all considered subgraphs are induced.
Let ${\cal G}$ be a set of graphs.
A {\it subgraph} of ${\cal G}$ is a graph $H$ such that, 
for every graph $G$ from ${\cal G}$,
the graph $G$ has an induced subgraph that is isomorphic to $H$.
A {\it supergraph} of ${\cal G}$ is a graph $H$ such that, 
for every graph $G$ from ${\cal G}$,
the graph $H$ has an induced subgraph that is isomorphic to $G$.
A subgraph that is a forest or tree, is called a {\it subforest} or {\it subtree}, respectively.
A supergraph that is a forest or tree, is called a {\it superforest} or {\it supertree}, respectively.

In this paper we consider the following natural optimization problems.

\medskip

\noindent {\sc Maximum Subforest}\\
\begin{tabular}{lp{14cm}}
Instance: & A set ${\cal F}$ of forests.\\
Task: & Determine a subforest $F$ of ${\cal F}$ of maximum order.
\end{tabular}

\medskip

\medskip

\noindent {\sc Minimum Superforest}\\
\begin{tabular}{lp{14cm}}
Instance: & A set ${\cal F}$ of forests.\\
Task: & Determine a superforest $F$ of ${\cal F}$ of minimum order.
\end{tabular}

\medskip

\medskip

Both problems are already NP-hard restricted to instances ${\cal F}=\{ F_1,F_2\}$,
where $F_1$ and $F_2$ are unions of paths: 
Let $I$ be an instance of the strongly NP-complete problem {\sc 3-partition}, 
cf.~[SP15] in~\cite{gajo}.
Let $I$ consist of $3m$ positive integers $a_1,\ldots,a_{3m}$ 
with $A/4<a_i<A/2$ for each $i\in [3m]$,
where $A=\frac{1}{m}(a_1+\cdots+a_{3m})$.
The task for $I$ is to decide whether there is a partition of $[3m]$
into $m$ sets $I_1,\ldots,I_m$ each containing exactly three elements
such that $\sum_{j\in I_i}a_j=A$ for each $i\in [m]$.
Let $F_1$ be the forest with $3m$ components that are paths of order $a_1,\ldots,a_{3m}$
and let $F_2$ be the forest with $m$ components that are paths of order $A+2$.
Note that $n(F_2)=n(F_1)+2m$.

\pagebreak

Obviously, the following statements are equivalent:
\begin{enumerate}[(i)]
\item $I$ is a yes-instance of {\sc 3-partition}.
\item $F_1$ is isomorphic to an induced subtree of $F_2$.
\item $F_1$ is a maximum subforest of $\{F_1,F_2\}$.
\item $F_2$ is a minimum superforest of $\{F_1,F_2\}$.
\end{enumerate}
These equivalences imply the stated hardness of {\sc Maximum Subforest} 
and {\sc Minimum Superforest}.
They also show that these problems are closely related to the very well-studied 
{subtree}/{subgraph isomorphism} problem~\cite{ma,abba,boha,heva,mapi}.
Maximum common (induced and non-induced) subgraphs
were first studied by Bokhari~\cite{bo} in the context of array processing
and are applied in areas 
ranging from molecular chemistry~\cite{rawi} to pattern matching~\cite{shbuve}.
The maximum common connected induced subgraph problem 
was shown to be NP-hard 
for $3$-outerplanar labeled graphs 
of maximum degree and treewidth at most $4$~\cite{akmeta,akta}
and for two biconnected series-parallel graphs~\cite{krkumu}.
It can be solved efficiently~\cite{yaaoma} 
for a degree-bounded partial $k$-tree and a connected graph, 
whose number of spanning trees is polynomial.
For the maximum common induced subgraph problem 
the parameterized complexity is studied in \cite{ab,abbosi}.

Modifying the above NP-hardness comments similarly as in~\cite{grrawo} yields the following.

\begin{proposition}\label{proposition1}
{\sc Maximum Subforest} restricted to instances $\{ T_1,T_2\}$ consisting of two subdivided stars is NP-hard.
\end{proposition}
Note that all proofs are postponed to Section \ref{sec2}.

If the set ${\cal F}$ contains only trees and $F$ is a minimum superforest of ${\cal F}$,
then each copy of a tree from ${\cal F}$ is completely contained in one component of $F$.
If $F$ would not be connected, then selecting one vertex from each component of $F$ and 
identifying all selected vertices to a single vertex
would yield a strictly smaller superforest of ${\cal F}$.
This argument implies the following. 
\begin{eqnarray}\label{e1}
\mbox{\it Every minimum superforest of a set of trees is a tree.}
\end{eqnarray}
For two trees $T_1$ and $T_2$,
minimum supertree $T_{\cup}$ of $\{ T_1,T_2\}$, and
a maximum subtree $T_{\cap}$ of $\{ T_1,T_2\}$,
the following inclusion-exclusion formula 
concerning the orders of these trees 
is straightforward.
\begin{eqnarray}\label{e2}
n(T_{\cup})=n(T_1)+n(T_2)-n(T_{\cap}).
\end{eqnarray}
Furthermore, given subtrees of $T_1$ and $T_2$ isomorphic to $T_{\cap}$,
a minimum superforest of $\{ T_1,T_2\}$ can easily be constructed 
by extending the copy of $T_{\cap}$ within $T_1$ 
by adding $n(T_2)-n(T_\cap)$ new vertices and suitable edges
creating a copy of $T_2$.
Refering to Edmonds and Matula, 
Akutsu~\cite{ak} showed that, for two given trees $T_1$ and $T_2$, 
some maximum subtree of $\{ T_1,T_2\}$ can be determined efficiently 
combining a weighted bipartite matching algorithm with dynamic programming.

Together our comments imply the following.

\begin{proposition}\label{proposition2}
{\sc Minimum Superforest} restricted to instances $\{ T_1,T_2\}$ consisting of two trees 
can be solved in polynomial time.
\end{proposition}
In \cite{ak} Akutsu also showed that it is NP-hard to determine
a maximum subtree of three given trees.
Reflecting this result, we show the following,
which does not follows from Akutsu's result.

\begin{theorem}\label{theorem1}
{\sc Minimum Superforest} restricted to instances $\{ T_1,T_2,T_3\}$ consisting of three trees is NP-hard.
\end{theorem}
For instances of bounded maximum degree, the problem can be solved efficiently.

\begin{theorem}\label{theorem5}
For every $\Delta\in \mathbb{N}$, 
there is some $p\in \mathbb{N}$ with the following property:
For a given set ${\cal T}=\{ T_1,T_2,T_3\}$ consisting of $3$ trees 
of order at most $n$ and maximum degree at most $\Delta$,
one can determine in time $O(n^p)$ 
a minimum superforest $T$ of ${\cal T}$.
\end{theorem}
By Proposition \ref{proposition2}, 
some minimum supertree, say $s(T,T')$, of two given trees $T$ and $T'$
can be determined efficiently.
Repeated applications of this lead to the following natural simple greedy algorithm. 

\medskip

\medskip

\begin{algorithm}[htb!]
\SetAlgoLined
\KwIn{A set $\{ T_1,\ldots,T_k\}$ of trees.}
\KwOut{A supertree $T$ of $\{ T_1,\ldots,T_k\}$.}
\Begin{
\For{$i=1$ \KwTo $k$}
{
$S_i\leftarrow T_i$\;
\For{$j=2$ \KwTo $k$}
{
$S_i\leftarrow s(S_i,T_{i+j-1})$, where indices are identified modulo $k$\;
}
}
$\ell\leftarrow {\rm argmin}\{ n(S_i):i\in [k]\}$\;
\Return{$S_\ell$}\;
}
\caption{\texttt{Greedy Supertree}}
\end{algorithm}

\begin{theorem}\label{theorem2}
\texttt{Greedy Supertree} is an efficient 
$\left(\frac{k}{2}-\frac{1}{2}+\frac{1}{k}\right)$-approximation algorithm 
for {\sc Minimum Superforest} restricted to instances $\{ T_1,\ldots,T_k\}$ consisting of $k$ trees.
\end{theorem}
For $k=3$, Theorem \ref{theorem2} provides the approximation factor $4/3$.
In Section \ref{sec2} we show that our analysis of \texttt{Greedy Supertree} 
is essentially best possible and that this factor can not be improved.
The appearance of the factor $4/3$ in this context is actually not surprising.
A natural simple dynamic programming algorithm 
that determines a minimum supertree of two given trees
uses a maximum bipartite matching algorithm as a subroutine.
Extending this dynamic programming approach from two to three trees 
would require to replace this subroutine with a $3$-dimensional matching algorithm.
Now, $4/3+\epsilon$ is the best known approximation factor for $3$-dimensional matching~\cite{cy}
with no improvement during the past decade.
More generally, 
the approximation factor in Theorem \ref{theorem2} reflects that the best known~\cite{husc} 
approximation factor for the $k$-set packing problem is $k/2+\epsilon$.
Altogether, a natural challenging problem in this context 
is to improve the approximation factor of $4/3$ 
for {\sc Minimum Superforest} for sets $\{ T_1,T_2,T_3\}$ of three given trees.

In contrast to that {\sc Maximum Subforest} allows a polynomial time approximation.

\begin{theorem}\label{theorem4}
For every $\epsilon>0$, there is some $p\in \mathbb{N}$ with the following property:
For a given set ${\cal F}=\{ F_1,\ldots,F_k\}$ consisting of $k$ forests of order at most $n$,
one can determine in time $O(k n^p)$ 
a subforest $F$ of ${\cal F}$ with
$n(F)\geq (1-\epsilon)n(F_{\rm opt})$,
where $F_{\rm opt}$ is some maximum subforest of ${\cal F}$.
\end{theorem}

\section{Proofs}\label{sec2}

\begin{proof}[Proof of Proposition \ref{proposition1}]
Let $I$ be an instance of {\sc 3-partition}
that consists of $3m$ positive integers $a_1,\ldots,a_{3m}$ 
with $A/4<a_i<A/2$ for each $i\in [3m]$,
where $A=\frac{1}{m}(a_1+\cdots+a_{3m})$.
Let $F_1$ be the forest with $3m$ components that are paths of order $a_1,\ldots,a_{3m}$
and let $F_2$ be the forest with $m$ components that are paths of order $A+2$.
 Let $T_1$ arise from $F_1$ by adding one new vertex $r_1$ as well as 
 $3m$ new edges between $r_1$ and one endvertex in each component of $F_1$.
Similarly, let $T_2$ arise from $F_2$ by adding one new vertex $r_2$ as well as 
$m$ new edges between $r_2$ and one endvertex in each component of $F_2$.
Note that $T_1$ and $T_2$ are subdivided stars.

In order to complete the proof, we show that $I$ is a yes-instance of {\sc 3-partition} 
if and only if 
a maximum subforest of $\{ T_1,T_2\}$ has order $n(T_1)-1=n(F_1)=a_1+\cdots+a_{3m}$.
Clearly, we may assume that $m\geq 8$.
Note that, since $T_2$ contains no vertex of degree $d_{T_1}(r_1)=3m$, 
$T_1$ is not a subtree of $T_2$, and, hence, 
a maximum subforest of $\{ T_1,T_2\}$ has order at most $n(T_1)-1$.

If $I$ is a yes-instance of {\sc 3-partition}, then 
removing from $T_1$ only the vertex $r_1$ and
removing from $T_2$ the vertex $r_2$ as well as two further vertices from each component of $F_2$ corresponding to a feasible solution for $I$
yields two forests that are both isomorphic to $F_1$.
Conversely, suppose now that $F$ is an induced subforest of $T_1$ 
of order $n(T_1)-1$ 
that is isomorphic to an induced subforest $F'$ of $T_2$.
Note that $F$ arises from $T_1$ by removing a single vertex.
Suppose, for a contradiction, that $r_1$ belongs to $F$.
This implies $d_F(r_1)\geq d_{T_1}(r_1)-1=3m-1>m$.
Since $m$ is the maximum degree of $T_2$, this is impossible,
which implies $F=T_1-r_1=F_1$.
Suppose, for a contradiction, that $r_2$ belongs to $F'$.
Since $F$ is the union of paths, this implies that there are $m-2$ neighbors $u_1,\ldots,u_{m-2}$ of $r_2$ in $T_2$ that do not belong to $F'$.
Let $P_1,\ldots,P_{m-2}$ be the components of $F_2$ such that $P_i$ contains $u_i$ for $i\in [m-2]$.
Since each $P_i-u_i$ is a path of order $A+1$ and each $a_i$ is strictly less than $A/2$,
for each $P_i$, there are at least three vertices that do not belong to $F'$.
Since $m\geq 8$, this implies the contradiction $n(F')\leq n(T_2)-3(m-2)<n(T_2)-2m-1=n(T_1)-1=n(F_1)=n(F)$.
Hence, $r_2$ does not belong to $F'$, that is, $F'$ is an induced subforest of $F_2$.
Again, since each component of $F_2$ is a path of order $A+2$ and each $a_i$ is strictly less than $A/2$,
for each component of $F_2$, there are at least two vertices that do not belong to $F'$.
Since $n(F')=n(F)=n(T_2)-2m-1=n(F_2)-2m$,
it follows that each component of $F_2$ contains exactly two vertices that do not belong to $F'$.
These two vertices from each component of $F_2$ indicate a feasible solution for $I$,
which implies that $I$ is a yes-instance of {\sc 3-partition}.
\end{proof}

\begin{proof}[Proof of Theorem \ref{theorem1}]
We show this result by an efficient reduction of the well-known 
NP-complete problem {\sc $3$-dimensional matching} (3DM), cf.~[SP1] in~\cite{gajo},
to {\sc Minimum Superforest}.
Let $I$ be an instance of 3DM consisting of three disjoint sets
$X=\{ x_1,\ldots,x_q\}$,
$Y=\{ y_1,\ldots,y_q\}$, and
$Z=\{ z_1,\ldots,z_q\}$
as well as a set $M\subseteq X\times Y\times Z$ of triples.
As 3DM remains NP-complete under this restriction~\cite{gajo}, 
we assume that 
every element of $X\cup Y\cup Z$ occurs in some triple but
no element of $X\cup Y\cup Z$ occurs in more than three triples.

For each $i\in [q]$, 
let $T_0(x_i)$ be the tree that arises from the disjoint union of 
an isolated vertex $r(x_i)$ and three paths $P_1$, $P_2$, and $P_3$,
each of order $2q$,
by adding an edge between $r(x_i)$ and an endvertex of each $P_\ell$.
The three vertices in $T_0(x_i)$ at distance $j$ from $r(x_i)$ are associated with $y_j$ and 
the three vertices in $T_0(x_i)$ at distance $p+k$ from $r(x_i)$ are associated with $z_k$.

If $x_i$ is contained in three triples from $M$,
then the tree $T(x_i)$ arises from $T_0(x_i)$ 
by associating each triple $(x_i,y_j,z_k)$ from $M$ containing $x_i$ 
with a different path $P_\ell$ and 
attaching one new endvertex to each of the two vertices in that $P_\ell$ 
at distances $j$ and $q+k$ from $r(x_i)$,
that is, the two vertices associated with $y_j$ and $z_k$,
respectively. 
See Figure \ref{fig4} for an illustration.

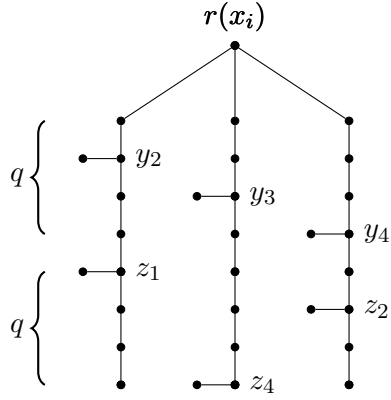
\begin{figure}[htb!]
    \centering
    \begin{tikzpicture}[scale=0.5]

    \tikzstyle{convcols}=[color=black]
    \tikzstyle{point}=[draw,circle,inner sep=0.cm, minimum size=1mm, fill=black]
    \tikzstyle{point2}=[draw,circle,inner sep=0.cm, minimum size=0.5mm, fill=black]
    \tikzstyle{pointred}=[draw,circle,inner sep=0.cm, minimum size=1.5mm, fill=red]

  \begin{scope}[yshift=12cm]

  \begin{scope}[xshift=-10cm]

\foreach \i in {1,2,3}{
            \node[point] (r1) at (0,0) [label=above:$r(x_i)$] {};
            \node[point] (a) at (3*\i-6,-9) [label=above:] {};
             \node[point] at (3*\i-6,-6) [label=above:] {};
            \node[point] (b) at (3*\i-6,-2) [label=above:] {};
            \draw (b) -- (0,0);
            \draw (b) -- (a);
    }
            \node[point] (x) at (-3,-3) [label=right:$y_2$] {};
            \node[point] (y) at (-4,-3) [label=above:] {};
            \draw (x) -- (y);

            \node[point] (x) at (0,-4) [label=right:$y_3$] {};
            \node[point] (y) at (-1,-4) [label=above:] {};
            \draw (y) -- (x);

            \node[point] (x) at (3,-5) [label=right:$y_4$] {};
            \node[point] (y) at (2,-5) [label=above:] {};
            \draw (y) -- (x);

    \draw[thick,black,decorate,decoration={brace,amplitude=5}] (-5,-5) -- (-5,-2)   node[midway,left]{$q\text{ }$};
    \draw[thick,black,decorate,decoration={brace,amplitude=5}] (-5,-9) -- (-5,-6)   node[midway,left]{$q\text{ }$};

                        \node[point] (x) at (-3,-6) [label=right:$z_1$] {};
            \node[point] (y) at (-4,-6) [label=above:] {};
            \draw (x) -- (y);

            \node[point] (x) at (0,-9) [label=right:$z_4$] {};
            \node[point] (y) at (-1,-9) [label=above:] {};
            \draw (y) -- (x);

            \node[point] (x) at (3,-7) [label=right:$z_2$] {};
            \node[point] (y) at (2,-7) [label=above:] {};
            \draw (y) -- (x);

\node[point] (y) at (-3,-4) [label=above:] {};
\node[point] (y) at (-3,-5) [label=above:] {};
\node[point] (y) at (-3,-6) [label=above:] {};
\node[point] (y) at (-3,-7) [label=above:] {};
\node[point] (y) at (-3,-8) [label=above:] {};
\node[point] (y) at (-3,-9) [label=above:] {};
  
\node[point] (y) at (0,-3) [label=above:] {};
\node[point] (y) at (0,-4) [label=above:] {};
\node[point] (y) at (0,-5) [label=above:] {};
\node[point] (y) at (0,-6) [label=above:] {};
\node[point] (y) at (0,-7) [label=above:] {};
\node[point] (y) at (0,-8) [label=above:] {};
\node[point] (y) at (0,-9) [label=above:] {};

\node[point] (y) at (3,-3) [label=above:] {};
\node[point] (y) at (3,-4) [label=above:] {};
\node[point] (y) at (3,-5) [label=above:] {};
\node[point] (y) at (3,-6) [label=above:] {};
\node[point] (y) at (3,-7) [label=above:] {};
\node[point] (y) at (3,-8) [label=above:] {};
\node[point] (y) at (3,-9) [label=above:] {};

            \end{scope}
  \end{scope}
\end{tikzpicture}
\caption{The tree $T(x_i)$ if $q=4$ and $x_i$ is contained in the three triples 
$(x_i,y_2,z_1)$, $(x_i,y_3,z_4)$, and $(x_i,y_4,z_2)$.}\label{fig4}
\end{figure}
If $x_i$ is contained in less than three triples, 
then proceed as before for the one or two triples containing $x_i$
and attach a new endvertex to each of the $2q$ vertices of those $P_\ell$
that are not associated with some triple containing $x_i$.
See Figure \ref{fig5} for an illustration.

\begin{figure}[htb!]
    \centering
    \begin{tikzpicture}[scale=0.5]

    \tikzstyle{convcols}=[color=black]
    \tikzstyle{point}=[draw,circle,inner sep=0.cm, minimum size=1mm, fill=black]
    \tikzstyle{point2}=[draw,circle,inner sep=0.cm, minimum size=0.5mm, fill=black]
    \tikzstyle{pointred}=[draw,circle,inner sep=0.cm, minimum size=1.5mm, fill=red]

  \begin{scope}[yshift=12cm]

  \begin{scope}[xshift=-10cm]

\foreach \i in {1}{
            \node[point] (r1) at (0,0) [label=above:$r(x_i)$] {};
            \node[point] (a) at (3*\i-6,-9) [label=above:] {};
             \node[point] at (3*\i-6,-6) [label=above:] {};
            \node[point] (b) at (3*\i-6,-2) [label=above:] {};
            \draw (b) -- (0,0);
            \draw (b) -- (a);

                        \foreach \j in {-2,-3,-4,-5,-6,-7,-8,-9}{
            \node[point] (a) at (3*\i-6,\j) [label=above:] {};
            \node[point] (b) at (3*\i-7,\j) [label=above:] {};
         \draw (b) -- (a);
         }

    }

\foreach \i in {2,3}{
            \node[point] (r1) at (0,0) [label=above:$r(x_i)$] {};
            \node[point] (a) at (3*\i-6,-9) [label=above:] {};
             \node[point] at (3*\i-6,-6) [label=above:] {};
            \node[point] (b) at (3*\i-6,-2) [label=above:] {};
            \draw (b) -- (0,0);
            \draw (b) -- (a);
    }

            \node[point] (x) at (-3,-3) [label=above:] {};
            \node[point] (y) at (-4,-3) [label=above:] {};
            \draw (x) -- (y);

            \node[point] (x) at (0,-4) [label=above:] {};
            \node[point] (y) at (-1,-4) [label=above:] {};
            \draw (y) -- (x);

            \node[point] (x) at (3,-5) [label=above:] {};
            \node[point] (y) at (2,-5) [label=above:] {};
            \draw (y) -- (x);

    \draw[thick,black,decorate,decoration={brace,amplitude=5}] (-5,-5) -- (-5,-2)   node[midway,left]{$q\text{ }$};
    \draw[thick,black,decorate,decoration={brace,amplitude=5}] (-5,-9) -- (-5,-6)   node[midway,left]{$q\text{ }$};

                        \node[point] (x) at (-3,-6) [label=above:] {};
            \node[point] (y) at (-4,-6) [label=above:] {};
            \draw (x) -- (y);

            \node[point] (x) at (0,-9) [label=above:] {};
            \node[point] (y) at (-1,-9) [label=above:] {};
            \draw (y) -- (x);

            \node[point] (x) at (3,-7) [label=above:] {};
            \node[point] (y) at (2,-7) [label=above:] {};
            \draw (y) -- (x);

\node[point] (y) at (0,-3) [label=above:] {};
\node[point] (y) at (0,-4) [label=above:] {};
\node[point] (y) at (0,-5) [label=above:] {};
\node[point] (y) at (0,-6) [label=above:] {};
\node[point] (y) at (0,-7) [label=above:] {};
\node[point] (y) at (0,-8) [label=above:] {};
\node[point] (y) at (0,-9) [label=above:] {};

\node[point] (y) at (3,-3) [label=above:] {};
\node[point] (y) at (3,-4) [label=above:] {};
\node[point] (y) at (3,-5) [label=above:] {};
\node[point] (y) at (3,-6) [label=above:] {};
\node[point] (y) at (3,-7) [label=above:] {};
\node[point] (y) at (3,-8) [label=above:] {};
\node[point] (y) at (3,-9) [label=above:] {};

            \end{scope}



  \end{scope}

    \begin{scope}[yshift=12cm]

  \begin{scope}[xshift=5cm]

\foreach \i in {1,2}{
            \node[point] (r1) at (0,0) [label=above:$r(x_i)$] {};
            \node[point] (a) at (3*\i-6,-9) [label=above:] {};
             \node[point] at (3*\i-6,-6) [label=above:] {};
            \node[point] (b) at (3*\i-6,-2) [label=above:] {};
            \draw (b) -- (0,0);
            \draw (b) -- (a);

                        \foreach \j in {-2,-3,-4,-5,-6,-7,-8,-9}{
            \node[point] (a) at (3*\i-6,\j) [label=above:] {};
            \node[point] (b) at (3*\i-7,\j) [label=above:] {};
         \draw (b) -- (a);
         }

    }

\foreach \i in {3}{
            \node[point] (r1) at (0,0) [label=above:] {};
            \node[point] (a) at (3*\i-6,-9) [label=above:] {};
             \node[point] at (3*\i-6,-6) [label=above:] {};
            \node[point] (b) at (3*\i-6,-2) [label=above:] {};
            \draw (b) -- (0,0);
            \draw (b) -- (a);
    }

            \node[point] (x) at (-3,-3) [label=above:] {};
            \node[point] (y) at (-4,-3) [label=above:] {};
            \draw (x) -- (y);

            \node[point] (x) at (0,-4) [label=above:] {};
            \node[point] (y) at (-1,-4) [label=above:] {};
            \draw (y) -- (x);

            \node[point] (x) at (3,-5) [label=above:] {};
            \node[point] (y) at (2,-5) [label=above:] {};
            \draw (y) -- (x);

    \draw[thick,black,decorate,decoration={brace,amplitude=5}] (-5,-5) -- (-5,-2)   node[midway,left]{$q\text{ }$};
    \draw[thick,black,decorate,decoration={brace,amplitude=5}] (-5,-9) -- (-5,-6)   node[midway,left]{$q\text{ }$};

                        \node[point] (x) at (-3,-6) [label=above:] {};
            \node[point] (y) at (-4,-6) [label=above:] {};
            \draw (x) -- (y);

            \node[point] (x) at (0,-9) [label=above:] {};
            \node[point] (y) at (-1,-9) [label=above:] {};
            \draw (y) -- (x);

            \node[point] (x) at (3,-7) [label=above:] {};
            \node[point] (y) at (2,-7) [label=above:] {};
            \draw (y) -- (x);

\node[point] (y) at (3,-3) [label=above:] {};
\node[point] (y) at (3,-4) [label=above:] {};
\node[point] (y) at (3,-5) [label=above:] {};
\node[point] (y) at (3,-6) [label=above:] {};
\node[point] (y) at (3,-7) [label=above:] {};
\node[point] (y) at (3,-8) [label=above:] {};
\node[point] (y) at (3,-9) [label=above:] {};

            \end{scope}



  \end{scope}

    \end{tikzpicture}
\caption{The left shows $T(x_i)$ if $q=4$ and $x_i$ is contained in exactly the two triples
$(x_i,y_3,z_4)$ and $(x_i,y_4,z_2)$. 
The right shows $T(x_i)$ if $q=4$ and 
$x_i$ is contained in only one triple $(x_i,y_4,z_2)$.}\label{fig5}
\end{figure}
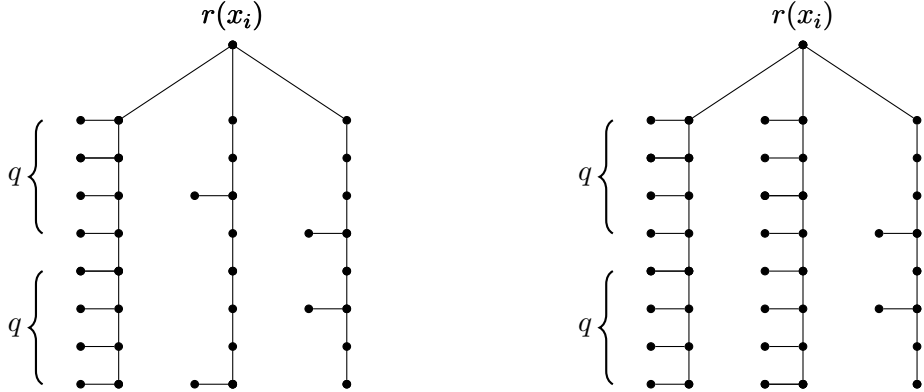
For each $j\in [q]$, 
let $T(y_j)$ be the tree that arises from the disjoint union of 
an isolated vertex $r(y_j)$ and three paths $P_1$, $P_2$, and $P_3$, 
each of order $2q$, by 
\begin{itemize}
\item adding an edge between $r(y_j)$ and an endvertex of each $P_\ell$,
\item attaching a new endvertex to each of the $2q$ vertices of two of the $P_\ell$, and
\item attaching one new endvertex to the vertex at distance $j$ from $r(y_j)$ 
on the third $P_\ell$,
which we call {\it the relevant branch for $y_j$}
in what follows.
\end{itemize}
Let $T(z_k)$ be defined similarly. 
In particular, $T(z_k)$ has an endvertex attached 
to a vertex at distance $q+k$ from $r(z_k)$; 
see Figure \ref{fig6} for an illustration.

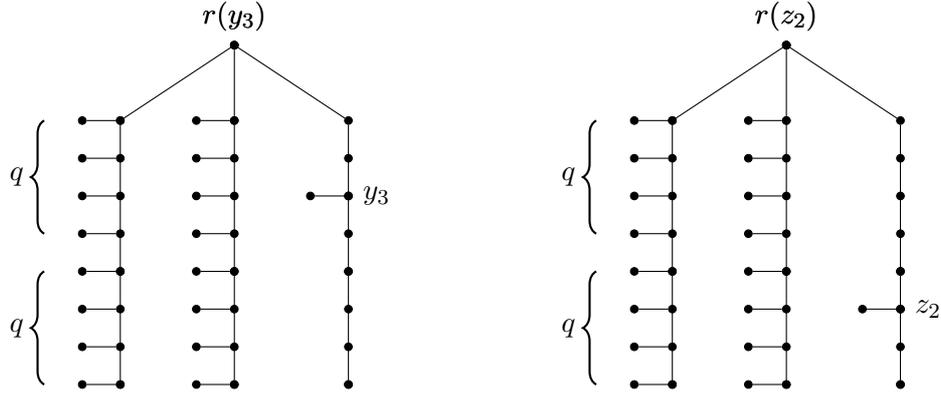
\begin{figure}[htb!]
    \centering
    \begin{tikzpicture}[scale=0.5]

    \tikzstyle{convcols}=[color=black]
    \tikzstyle{point}=[draw,circle,inner sep=0.cm, minimum size=1mm, fill=black]
    \tikzstyle{point2}=[draw,circle,inner sep=0.cm, minimum size=0.5mm, fill=black]
    \tikzstyle{pointred}=[draw,circle,inner sep=0.cm, minimum size=1.5mm, fill=red]

    \begin{scope}[yshift=12cm]

  \begin{scope}[xshift=5cm]

\foreach \i in {1,2}{
            \node[point] (r1) at (0,0) [label=above:$r(y_3)$] {};
            \node[point] (a) at (3*\i-6,-9) [label=above:] {};
             \node[point] at (3*\i-6,-6) [label=above:] {};
            \node[point] (b) at (3*\i-6,-2) [label=above:] {};
            \draw (b) -- (0,0);
            \draw (b) -- (a);

                        \foreach \j in {-2,-3,-4,-5,-6,-7,-8,-9}{
            \node[point] (a) at (3*\i-6,\j) [label=above:] {};
            \node[point] (b) at (3*\i-7,\j) [label=above:] {};
         \draw (b) -- (a);
         }

    }

\foreach \i in {3}{
            \node[point] (r1) at (0,0) [label=above:] {};
            \node[point] (a) at (3*\i-6,-9) [label=above:] {};
             \node[point] at (3*\i-6,-6) [label=above:] {};
            \node[point] (b) at (3*\i-6,-2) [label=above:] {};
            \draw (b) -- (0,0);
            \draw (b) -- (a);
    }

    \draw[thick,black,decorate,decoration={brace,amplitude=5}] (-5,-5) -- (-5,-2)   node[midway,left]{$q\text{ }$};
    \draw[thick,black,decorate,decoration={brace,amplitude=5}] (-5,-9) -- (-5,-6)   node[midway,left]{$q\text{ }$};

            \node[point] (x) at (3,-4) [label=right:$y_3$] {};
            \node[point] (y) at (2,-4) [label=above:] {};
            \draw (y) -- (x);

\node[point] (y) at (3,-3) [label=above:] {};
\node[point] (y) at (3,-4) [label=above:] {};
\node[point] (y) at (3,-5) [label=above:] {};
\node[point] (y) at (3,-6) [label=above:] {};
\node[point] (y) at (3,-7) [label=above:] {};
\node[point] (y) at (3,-8) [label=above:] {};
\node[point] (y) at (3,-9) [label=above:] {};

            \end{scope}



  \end{scope}

    \end{tikzpicture}\hspace{2cm}\begin{tikzpicture}[scale=0.5]

    \tikzstyle{convcols}=[color=black]
    \tikzstyle{point}=[draw,circle,inner sep=0.cm, minimum size=1mm, fill=black]
    \tikzstyle{point2}=[draw,circle,inner sep=0.cm, minimum size=0.5mm, fill=black]
    \tikzstyle{pointred}=[draw,circle,inner sep=0.cm, minimum size=1.5mm, fill=red]

    \begin{scope}[yshift=12cm]

  \begin{scope}[xshift=5cm]

\foreach \i in {1,2}{
            \node[point] (r1) at (0,0) [label=above:$r(z_2)$] {};
            \node[point] (a) at (3*\i-6,-9) [label=above:] {};
             \node[point] at (3*\i-6,-6) [label=above:] {};
            \node[point] (b) at (3*\i-6,-2) [label=above:] {};
            \draw (b) -- (0,0);
            \draw (b) -- (a);

                        \foreach \j in {-2,-3,-4,-5,-6,-7,-8,-9}{
            \node[point] (a) at (3*\i-6,\j) [label=above:] {};
            \node[point] (b) at (3*\i-7,\j) [label=above:] {};
         \draw (b) -- (a);
         }

    }

\foreach \i in {3}{
            \node[point] (r1) at (0,0) [label=above:] {};
            \node[point] (a) at (3*\i-6,-9) [label=above:] {};
             \node[point] at (3*\i-6,-6) [label=above:] {};
            \node[point] (b) at (3*\i-6,-2) [label=above:] {};
            \draw (b) -- (0,0);
            \draw (b) -- (a);
    }

    \draw[thick,black,decorate,decoration={brace,amplitude=5}] (-5,-5) -- (-5,-2)   node[midway,left]{$q\text{ }$};
    \draw[thick,black,decorate,decoration={brace,amplitude=5}] (-5,-9) -- (-5,-6)   node[midway,left]{$q\text{ }$};

            \node[point] (x) at (3,-7) [label=right:$z_2$] {};
            \node[point] (y) at (2,-7) [label=above:] {};
            \draw (y) -- (x);
\node[point] (y) at (3,-3) [label=above:] {};
\node[point] (y) at (3,-4) [label=above:] {};
\node[point] (y) at (3,-5) [label=above:] {};
\node[point] (y) at (3,-6) [label=above:] {};
\node[point] (y) at (3,-7) [label=above:] {};
\node[point] (y) at (3,-8) [label=above:] {};
\node[point] (y) at (3,-9) [label=above:] {};

            \end{scope}



  \end{scope}

    \end{tikzpicture}
    
    \caption{$T(y_3)$ on the left and $T(z_2)$ on the right.}\label{fig6}
\end{figure}
Now, let $T_x$ arise from the disjoint union of an isolated vertex $r_x$
and the trees $T(x_1),\ldots,T(x_q)$ by adding $q$ new edges between $r_x$
and $r(x_1),\ldots,r(x_q)$.
Let $T_y$ and $T_z$ be defined similarly.
Note that the trees $T_x$, $T_y$, and $T_z$ 
are rooted in the vertices $r_x$, $r_y$, and $r_z$ of degree $q$, respectively.
The order of $T_y$ and $T_z$ is $n=10q^2+2q+1$ 
while the order of $T_x$ depends on the instance $I$.
In order to complete the proof, we show that $I$ is a yes-instance of 3DM
if and only if a minimum superforest for $\{ T_x,T_y,T_z\}$ 
has order at most $10q^2+3q+1$.

Suppose that $I$ is a yes-instance of 3DM.
Let $M^*\subseteq M$ be such that 
every element of $X\cup Y\cup Z$ 
belongs to exactly one triple from $M^*$.
Let the tree $T$ arise from $T_y$ as follows:
For each $j\in [q]$, consider the unique triple, say $(x_i,y_j,z_k)$, 
from $M^*$ that contains $y_j$,
and attach a new endvertex to the vertex at distance $q+k$ 
from $r(y_j)$ 
associated with $z_k$ 
that belongs to the relevant branch for $y_j$.
Clearly, the order of $T$ is $n(T_y)+q=10q^2+3q+1$
and it is easy to verify that $T$ contains three induced subtrees 
isomorphic to $T_x$, $T_y$, and $T_z$, respectively.

Conversely, suppose that a minimum superforest $T$ for $\{ T_x,T_y,T_z\}$ 
has order at most $10q^2+3q+1$, which equals $n(T_y)+q$.
By renaming vertices, we may assume that $T$ arises from $T_y$
by adding at most $q$ vertices and suitable edges.
The structure of $T_y$ and $T_z$ implies that $T$ arises from $T_y$
by attaching one new endvertex to some vertex of each of the $q$ relevant branches within $T_y$; these $q$ additional vertices are attached to vertices 
associated with the distinct elements of $Z$.
Since $T_x$ is an induced subgraph of $T$
and $r_y$ is the only vertex of $T$ of degree $q$,
for a copy of $T_x$ within $T$, the root vertex $r_x$ of $T_x$ is mapped to $r_y$
and the $q$ children of $r_x$ in $T_x$ are mapped in a bijective way to the $q$ children of $r_y$ within $T$.
This bijective mapping indicates how to choose, for each $i\in [q]$,
a triple from $M$ containing $x_i$,
for which the set $M^*$ of all $q$ selected triples 
is such that every element of $X\cup Y\cup Z$ 
is contained in exactly one triple from $M^*$.
This completes the proof.
\end{proof}

\begin{proof}[Proof of Theorem \ref{theorem5}]
Let ${\cal T}=\{ T_1,T_2,T_3\}$ be the set of the three given trees 
of order at most $n$ and maximum degree at most $\Delta$.
For notational simplicity, assume that the trees $T_i$ have disjoint sets of vertices. 
We explain how to determine in polynomial time 
supertrees $T$ of ${\cal T}$ that
\begin{itemize}
\item either contain disjoint copies of two of the three trees (type 1) 
\item or contain copies of all three trees that pairwise intersect (type 2)
\end{itemize}
and are of minimum order subject to this condition.
Returning the smallest such supertree yields a minimum supertree of ${\cal T}$.

Firstly, consider supertrees of type 1 
that contain disjoint copies of $T_2$ and $T_3$;
the other two pairs can be treated symmetrically.
Let $\{ T_4,\ldots,T_q\}$ be the set of all trees 
that arise from disjoint copies of $T_2$ and $T_3$ and 
a path $P$ of order between $2$ and $n$ by identifying 
some vertex $u$ in $T_2$ with one endvertex of $P$
and some vertex $v$ of $T_3$ with the other endvertex of $P$.
Since there are at most $n$ choices 
for the length of $P$, 
for the vertex $u$, and 
for the vertex $v$,
we have $q=O(n^3)$. 
By Proposition \ref{proposition2},
the $O(n^3)$ minimum supertrees of 
$\{ T_1,T_4\},\{ T_1,T_5\},\ldots,\{ T_1,T_q\}$
can be determined in polynomial time,
and a smallest of all these trees is a supertree of ${\cal T}$
containing disjoint copies of $T_2$ and $T_3$
that is of minimum order subject to this condition.

Secondly, consider a supertree $T$ of type 2.
Let $T_i'$ be an induced copy of $T_i$ within $T$
such that $T_1'$, $T_2'$, and $T_3'$ pairwise intersect.
By the Helly property of subtrees of a tree, 
some vertex, say $r$, belongs to $T_1'$, $T_2'$, and $T_3'$.
For all possible at most $O(n^3)$ choices for vertices 
$r_1$ in $T_1$,
$r_2$ in $T_2$, and 
$r_3$ in $T_3$ corresponding to $r$, we proceed as follows for every $i\in [3]$:
\begin{itemize}
\item Root $T_i$ in $r_i$.
\item For $u_i\in V(T_i)$, let $T_i(u_i)$ be the subtree of $T_i$ rooted in $u_i$ 
that is induced by $u_i$ and all its descendants within $T_i$.
\item Let $f(\{ u_i\})=n(T_i(u_i))$.
\item Let $f(\{ u_1,u_2,u_3\})$ be the minimum order of a supertree $T$ 
of $\{ T_1(u_1),T_2(u_2),T_3(u_3)\}$ rooted in some vertex $s$ such that 
$T$ contains a copy of $T_i$ in which $s$ corresponds to $u_i$ for every $i\in [3]$.
\item Let $f(\{ u_1,u_2\})$ be the minimum order of a supertree $T$ 
of $\{ T_1(u_1),T_2(u_2)\}$ rooted in some vertex $s$ such that 
$T$ contains a copy of $T_i$ in which $s$ corresponds to $u_i$ for every $i\in [2]$.
Define $f(\{ u_1,u_3\})$ and $f(\{ u_2,u_3\})$ symmetrically.
By Proposition \ref{proposition2}, 
$f(\{ u_1,u_2\})$,
$f(\{ u_1,u_3\})$, and
$f(\{ u_2,u_3\})$ can be determined efficiently.
\end{itemize}
Note that $f(\{ r_1,r_2,r_3\})$ is the minimum order of a supertree $T$ 
of $\{ T_1,T_2,T_3\}$ rooted in some vertex $r$ that contains 
a copy of $T_i$ in which $r$ corresponds to $r_i$ for every $i\in [3]$.
Since we consider all $O(n^3)$ choices for the $r_i$,
the smallest such tree is a minimum supertree of type 2.

In order to complete the proof, 
we explain how to determine the values $f(\{ u_1,u_2,u_3\})$
by dynamic programming in polynomial time.
If $u_1$ is an endvertex of $T_1$, then $f(\{ u_1,u_2,u_3\})=f(\{ u_2,u_3\})$;
similarly, if $u_2$ or $u_3$ are endvertices.
Hence, we may assume that $u_1$, $u_2$, and $u_3$ are no endvertices.
Let $U_i$ be the set of children of $u_i$ in $T_i$.
The definitions imply that $f(\{ u_1,u_2,u_3\})$ is the minimum of 
$$f({\cal P})=1+\sum\limits_{j=1}^kf(e_j)$$
over all partitions ${\cal P}=\{ e_1,\ldots,e_k\}$ of $U_1\cup U_2\cup U_3$
into sets $e_j$ with $|e_j\cap U_i|\leq 1$ for every $i\in [3]$ and $j\in [k]$.
Since $|U_1\cup U_2\cup U_3|\leq 3\Delta$,
there are finitely many such partitions.
Altogether, it follows that the values $f(\cdot)$
(together with suitable realizers) 
can be determined efficiently by dynamic programming,
which completes the proof.
\end{proof}
By an inductive argument also considering type 1 and type 2 supertrees
and using the Helly property, Theorem \ref{theorem5} easily generalizes 
to {\sc Minimum Supertree} for given sets of $k$ trees
with the polynomial bounding the running time depending on $k$.
Furthermore, Theorem \ref{theorem5} remains true 
under the weaker hypothesis that only two of the trees in ${\cal T}=\{ T_1,T_2,T_3\}$ 
have maximum degree $O(\log(n))$
and the third tree is of arbitrary maximum degree.

\begin{proof}[Proof of Theorem \ref{theorem2}]
Let ${\cal T}=\{ T_1,\ldots,T_k\}$ be the given set of $k$ trees.
By Proposition \ref{proposition2},
for two given trees $T$ and $T'$, 
some minimum supertree $s(T,T')$ of $\{ T,T'\}$
can be found efficiently.
The trees $S_i$ determined by \texttt{Greedy Supertree}
are of the form
\begin{eqnarray*}
S_i &=& 
s(\ldots s(s(s(T_i,T_{i+1}),T_{i+2}),T_{i+3})\ldots,T_{i+k-1})
\mbox{ for $i\in [k]$},
\end{eqnarray*}
where indices are identified modulo $k$.
We show that returning the smallest of the $S_i$ yields a 
$\left(\frac{k}{2}-\frac{1}{2}+\frac{1}{k}\right)$-approximation algorithm 
for {\sc Minimum Superforest} on ${\cal T}$.
Therefore, let $T$ be a minimum superforest of ${\cal T}$.
Let $n=n(T)$.
For $i\in [k]$, let $n_i=n(T_i)$ and let $V_i\subseteq V(T)$ 
be such that $T[V_i]\simeq T_i$.
For $ij\in {[k]\choose 2}$,
let $n_{ij}=|V_i\cap V_j|$.
Clearly,
\begin{eqnarray*}
n_i+n_j-n_{ij} & \leq & n \mbox{ for every $ij\in {[k]\choose 2}$ and}\\
\sum\limits_{i\in [k]}n_i-\sum\limits_{ij\in {[k]\choose 2}}n_{ij}& \leq & n.
\end{eqnarray*}
Adding these ${k\choose 2}+1$ inequalities yields
\begin{eqnarray}\label{e4}
k\sum\limits_{i\in [k]}n_i-2\sum\limits_{ij\in {[k]\choose 2}}n_{ij}\leq \left({k\choose 2}+1\right)n,
\end{eqnarray}
and, hence,
\begin{eqnarray}
\frac{1}{k}\sum\limits_{j\in [k]}\left(\sum\limits_{i\in [k]}n_i-\sum\limits_{i\in [k]\setminus \{ j\}}n_{ij}\right)
& = & \sum\limits_{i\in [k]}n_i-\frac{2}{k}\sum\limits_{ij\in {[k]\choose 2}}n_{ij}\nonumber \\
& \stackrel{(\ref{e4})}{\leq} & \frac{1}{k}\left({k\choose 2}+1\right)n\nonumber \\
& = & \left(\frac{k}{2}-\frac{1}{2}+\frac{1}{k}\right)n.\label{e3}
\end{eqnarray}
Since $n_{i(i+1)}$ is the order of some possibly not largest subtree of $\{ T_i,T_{i+1}\}$, 
we have 
$n(s(T_i,T_{i+1}))\stackrel{(\ref{e2})}{\leq} n_i+n_{i+1}-n_{i(i+1)}.$
Since $n_{i(i+2)}$ is the order of some subtree of $\{ T_i,T_{i+2}\}$
and the tree $s(T_i,T_{i+1})$ contains a copy of $T_i$,
we have 
$n(s(s(T_i,T_{i+1}),T_{i+2})
\leq n(s(T_i,T_{i+1}))+n_{i+2}-n_{i(i+2)}
\leq n_i+n_{i+1}-n_{i(i+1)}+n_{i+2}-n_{i(i+2)}.$
Using that $n_{ij}$ is the order of some subtree of $\{ T_i,T_j\}$ 
and that each tree of the form 
$s(\ldots s(s(T_i,T_{i+1}),T_{i+2})\ldots,T_{i+\ell})$
for some $\ell$ contains a copy of $T_i$,
it now follows inductively that
\begin{eqnarray}\label{e5}
n(S_j)&\leq &\sum\limits_{i\in [k]}n_i-\sum\limits_{i\in [k]\setminus \{ j\}}n_{ij}.
\end{eqnarray} 
Altogether, we obtain
\begin{eqnarray*} 
\min\{ n(S_i):i\in [k]\}
&\leq &\frac{1}{k}\sum\limits_{j\in [k]}n(S_j)
\stackrel{(\ref{e5})}{\leq} \frac{1}{k}\sum\limits_{j\in [k]}\left(\sum\limits_{i\in [k]}n_i-\sum\limits_{i\in [k]\setminus \{ j\}}n_{ij}\right)
\stackrel{(\ref{e3})}{\leq} \left(\frac{k}{2}-\frac{1}{2}+\frac{1}{k}\right)n,
\end{eqnarray*} which completes the proof.
\end{proof}
The analysis of \texttt{Greedy Supertree}
is essentially best possible.
We give an example for $k=3$ showing that the factor $4/3$ 
can not be improved.
For non-negative integers $n_1,\ldots,n_p$,
let the tree $T(n_1,\ldots,n_p)$ arise from 
a path $P:u_1\ldots u_p$ of order $p$
by attaching, for every $i\in [p]$,
exactly $n_i$ new endvertices to $u_i$.
For positive integers $a$, $b$, and $c$ with $a>b>c\geq 1$, 
consider the three trees 
\begin{align*}
T_1 & = T(0,b,a,a,c,0),\\
T_2 & = T(0,b,0,a,0,0,0,a,0),\mbox{ and }\\
T_3 & = T(0,b,0,0,0,0,0,a,0,c,a,0)
\end{align*}
illustrated in Figure \ref{fig1}.

\begin{figure}[htb!]
    \centering
    \begin{tikzpicture} [scale=0.21]
    \tikzstyle{point}=[draw,circle,inner sep=0.cm, minimum size=0.75mm, fill=black]
    \tikzstyle{point2}=[draw,circle,inner sep=0.cm, minimum size=0.3mm, fill=black]

\xdef\j{0} 
\xdef\k{0}

\foreach \arr in {{0,4,6,6,3,0},{0,4,0,6,0,0,0,6,0},{0,4,0,0,0,0,0,6,0,3,6,0}}{

\xdef\j{\j-7} 
\xdef\k{\k-9} 

\begin{scope}[shift={(\k,\j)}]

\xdef\i{0} 

\foreach \w in \arr{

\ifthenelse{\w=0}{
    \node[point] at (\i,0) [label=above :] {};
}{

\begin{scope}[shift={(\i,0)}]
    
\begin{scope}[shift={(-0.5*\w,0)}]

\node[point] (c) at (0.5*\w,0) [label=above :] {};
\begin{scope}[shift={(0,-0.5)}] 
\draw [rounded corners, rounded corners=1mm] (0,-1)--(\w,-1)--(\w,-2)-- (0,-2) -- cycle;
\node[point] (a) at (0.5,-1.5) [label=above :] {};
\node[point] (b) at (\w-0.5,-1.5) [label=above :] {};
\draw (a) -- (c) -- (b);
\foreach \i in {0.35*\w,0.5*\w,0.65*\w}{
    \node[point2] at (\i,-1.5) [label=above :] {};
}
\ifthenelse{\w=3}{\node at (0.5*\w,-3) [label=above :] {$c$}}{};
\ifthenelse{\w=4}{\node at (0.5*\w,-3) [label=above :] {$b$}}{};
\ifthenelse{\w=6}{\node at (0.5*\w,-3) [label=above :] {$a$}}{};

\end{scope}
\end{scope}
\end{scope}
}
\xdef\i{\i+7} 
}
\draw (0,0) -- (\i-7,0);
\end{scope}
}
\end{tikzpicture}

\caption{Three trees $T_1$, $T_2$, and $T_3$.}\label{fig1}
\end{figure}
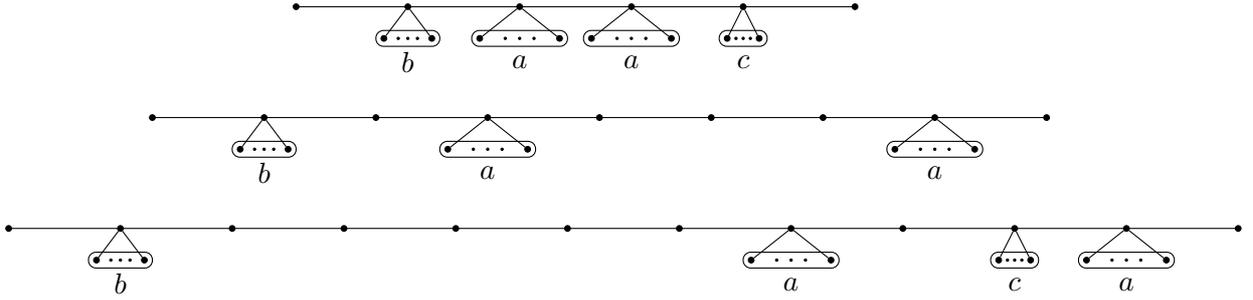

It is easy to verify that 
\begin{itemize}
\item $s(T_1,T_2)\simeq T(0,b,a,a,c,0,0,a,0)$ does not contain $T(0,a,0,0,a,0)$,
\item $s(T_1,T_3)\in \{ T(0,b,0,0,0,0,b,a,a,c,a,0),T(0,b,0,0,0,0,c,a,a,b,a,0)\}$ and does not contain $T(0,a,0,0,0,a,0)$, and
\item $s(T_2,T_3)\simeq T(0,b,0,a,0,0,0,a,0,c,a,0)$ does not contain $T(0,a,a,0)$.
\end{itemize}
It follows that all three trees
$s(s(T_1,T_2),T_3)$,
$s(s(T_1,T_3),T_2)$, and
$s(s(T_2,T_3),T_1)$
have order at least $4a$, while the tree 
$T(0,b,0,0,b,b,a,a,c,c,a,0)$
of order $3a+3b+2c+12$ 
shown in Figure \ref{fig3} 
contains $T_1$, $T_2$, and $T_3$.

\begin{figure}[htb!]
    \centering
    \begin{tikzpicture} [scale=0.21]
    \tikzstyle{point}=[draw,circle,inner sep=0.cm, minimum size=0.75mm, fill=black]
    \tikzstyle{point2}=[draw,circle,inner sep=0.cm, minimum size=0.3mm, fill=black]

\xdef\j{0} 
\xdef\k{0}

\foreach \arr in {{0,4,0,0,4,4,6,6,3,3,6,0}}{

\xdef\j{\j-7} 
\xdef\k{\k-9} 

\begin{scope}[shift={(\k,\j)}]

\xdef\i{0} 

\foreach \w in \arr{

\ifthenelse{\w=0}{
    \node[point] at (\i,0) [label=above :] {};
}{

\begin{scope}[shift={(\i,0)}]
    
\begin{scope}[shift={(-0.5*\w,0)}]

\node[point] (c) at (0.5*\w,0) [label=above :] {};
\begin{scope}[shift={(0,-0.5)}] 
\draw [rounded corners, rounded corners=1mm] (0,-1)--(\w,-1)--(\w,-2)-- (0,-2) -- cycle;
\node[point] (a) at (0.5,-1.5) [label=above :] {};
\node[point] (b) at (\w-0.5,-1.5) [label=above :] {};
\draw (a) -- (c) -- (b);
\foreach \i in {0.35*\w,0.5*\w,0.65*\w}{
    \node[point2] at (\i,-1.5) [label=above :] {};
}
\ifthenelse{\w=3}{\node at (0.5*\w,-3) [label=above :] {$c$}}{};
\ifthenelse{\w=4}{\node at (0.5*\w,-3) [label=above :] {$b$}}{};
\ifthenelse{\w=6}{\node at (0.5*\w,-3) [label=above :] {$a$}}{};

\end{scope}
\end{scope}
\end{scope}
}
\xdef\i{\i+7} 
}
\draw (0,0) -- (\i-7,0);
\end{scope}
}
\end{tikzpicture}
\caption{A supertree for $\{ T_1,T_2,T_3\}$.}\label{fig3}
\end{figure}
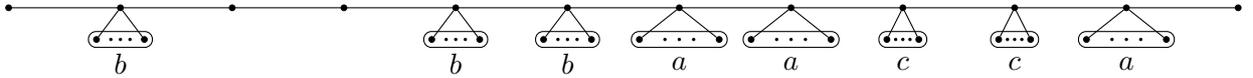
Choosing $a$ large shows that the factor $4/3$ can not be improved.
Note that every supertree of $\{ T_1,T_2\}$ that contains $T(0,a,0,0,a,0)$
and can therefore accomodate $T_3$ in a more efficient way
has at least $b$ vertices more than $s(T_1,T_2)$.
Choosing $a$, $b$, and $c$
such that $4\epsilon a\geq b>\epsilon n(s(T_1,T_2))=\epsilon(3a+b+c+9)$
shows that the factor $4/3$ can only be improved marginally
if the subroutine for $s(\cdot,\cdot)$ is allowed to return slightly suboptimal trees.

\pagebreak

For the proof of Theorem \ref{theorem4}, we need an auxiliary statement.

Let $\Delta\in \mathbb{N}$.
Let 
${\cal T}_{\Delta}=\{ T_1,\ldots,T_q\}$
be the set of all trees of order at most $\Delta$.
Let ${\cal F}_{\Delta}$ be the set of all forests whose components belong to ${\cal T}_{\Delta}$.
For a forest $F$, let 
$t(F)=(t_1,\ldots,t_q)\in [n(F)]_0^q$
be such that
$t_i$ is the number of components of $F$ that are isomorphic to $T_i$ for every $i\in [q]$ and
let
$$\hat{t}(F)=\{ t(F'):\mbox{$F'$ is an induced subforest of $F$}\}.$$
For every $t=(t_1,\ldots,t_q)$ from $\hat{t}(F)$, 
an induced subforest $F'$ of $F$ with $t(F')=t$
is a {\it realizer of $t$ within $F$}.
Note that $t(F)$ counts only small components of $F$ 
but that $F$ may have large components.
Note furthermore, that $\hat{t}(F)\subseteq [n(F)]^{q}_0$.

For two sets $A,B\in \mathbb{N}_0^q$, let $A\oplus B=\{ a+b:a\in A\mbox{ and }b\in B\}$.

\begin{lemma}\label{lemma1}
For every $\Delta\in \mathbb{N}$, there is some $p\in \mathbb{N}$ with the following property:
For every forest $F$ of order at most $n$, 
$\hat{t}(F)$ as well as realizers within $F$ 
can be determined in time $O(n^p)$.
\end{lemma}
\begin{proof}
If $F$ has components $F_1,\ldots,F_k$, then 
$\hat{t}(F)=\bigoplus\limits_{i=1}^k\hat{t}(F_i)$.
Since $\hat{t}(F_i)\subseteq [n(F_i)]_0^q$ and $\oplus$ is associative,
in order to show the desired statement, we may assume that $F$ is a tree.
Root $F$ is some vertex $r$.
Let ${\cal R}$ be the set of all pairs $(S,s)$ such that $S\in {\cal T}_{\Delta}$ and $s\in V(S)$,
that is, ${\cal R}$ captures all possible ways of selecting root vertices for the trees in ${\cal T}_{\Delta}$.
Let $u$ be some vertex of $F$.
Let $F_u$ be the subtree of $F$ rooted in $u$ that contains $u$ and all its descendants.

For every $(S,s)\in {\cal R}$, 
let $\hat{t}_{(S,s)}(F_u)$ be the set of all $(t_1,\ldots,t_q)\in [n(F_u)]_0^q$ such that 
\begin{itemize}
\item $F_u$ has an induced subforest $K$ that consists of $t_i$ disjoint copies of $T_i$ for every $i\in [q]$,
\item the vertex $u$ is contained in some component $L$ of $K$ that is isomorphic to $S$, and 
\item some isomorphism $\pi$ between $S$ and $L$ maps $s$ to $u$.
\end{itemize}
See Figure \ref{fign1} for an illustration.

Note that $\hat{t}_{(S,s)}(F_u)$ is empty, 
if $F_u$ does not contain a suitable copy of $S$.
\begin{figure}[H]
    \centering
    \begin{tikzpicture} [scale=0.45]
    \tikzstyle{point}=[draw,circle,inner sep=0.cm, minimum size=1mm, fill=black]
    \tikzstyle{point2}=[draw,circle,inner sep=0.cm, minimum size=0.5mm, fill=black]

    \tikzstyle{line1}=[line width=0.4mm]
    \tikzstyle{line2}=[thin]

\begin{scope}[shift={(-5,0)}]
     \draw [rounded corners, rounded corners=5mm] (-1.12,-2)--(5,8.6)--(11.12,-2)--cycle;    
\node at (5,7.75) [label=above :$F_u$] {};

\node[point] (v3) at (5,7.5) [label=above :] {};
\node at (5.25,7.5) [label=right :$u$] {};
\node[point] (v3r) at (5,13) [label=above :] {};
\node at (5,13) [label=right :$r$] {};

\end{scope}

\begin{scope}[shift={(2,-1)}]
    \xdef\www{16*0.4} 
     \xdef\hhh{13.9*0.4}     
    \begin{scope}[shift={(0,-0.5)}]
    \draw [rounded corners, rounded corners=10mm] (-\www/2,0)--(0,\hhh)--(\www/2,0)--cycle;
    \end{scope}
\node at (0,0.75) [label=above :$ \cong T_{k}$] {};
\end{scope}

\begin{scope}[shift={(-2.8,-1)}]
    \xdef\www{16*0.3} 
     \xdef\hhh{13.9*0.3}     
    \begin{scope}[shift={(0,-0.5)}]
    \draw [rounded corners, rounded corners=7.5mm] (-\www/2,0)--(0,\hhh)--(\www/2,0)--cycle;
    \end{scope}
\node at (0,0) [label=above :$ \cong T_i$] {};
\end{scope}

\begin{scope}[shift={(-0.75,4)}]
    \xdef\www{16*0.25} 
     \xdef\hhh{13.9*0.25}     
    \begin{scope}[shift={(0,-0.5)}]
    \draw [rounded corners, rounded corners=7.5mm] (-\www/2,0)--(0,-\hhh)--(\www/2,0)--cycle;
    \end{scope}
\node at (0,-2.75) [label=above :$ \cong T_j$] {};
\end{scope}

     \xdef\www{16*0.33} 
     \xdef\hhh{13.9*0.33} 
    
    \begin{scope}[shift={(0,10.6-\hhh)}]
    \draw [rounded corners, rounded corners=5mm] (-\www/2,-2)--(0,\hhh-2)--(\www/2,-2)--cycle;
    \end{scope}
\node[align=center] at (0,5.5) {$S\overset{\pi}{\cong} L$\\$\pi(s)=u$};

    \begin{scope}[shift={(0,0)}]
 \xdef\www{16*1.25} 
     \xdef\hhh{13.9*1.25}     
 \node at (1.5,\hhh-4) [label=right:$F$] {};        
    \draw [rounded corners, rounded corners=10mm] (-\www/2,-2)--(0,\hhh-2)--(\www/2,-2)--cycle;
    \end{scope}

\end{tikzpicture}

\caption{The structure of induced subgraphs $K$ of $F_u$ contributing to $\hat{t}_{(S,s)}(F_u)$.}\label{fign1}
\end{figure}
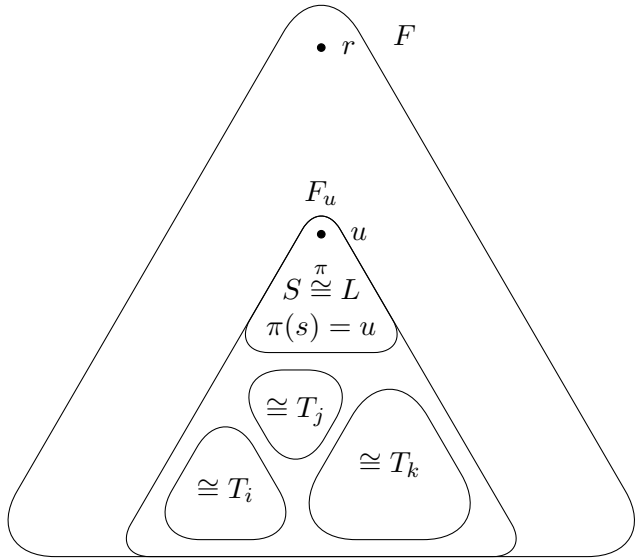

Similarly, 
let $\hat{t}_{\emptyset}(F_u)$ be the set of all $(t_1,\ldots,t_q)\in [n(F_u)]_0^q$ such that 
\begin{itemize}
\item $F_u$ has an induced subforest $K$ that consists of $t_i$ disjoint copies of $T_i$ for every $i\in [q]$ and
\item the vertex $u$ does not belong to $K$.
\end{itemize}
Clearly, 
$$\hat{t}(F_u)=\hat{t}_{\emptyset}(F_u)\cup \bigcup\limits_{(S,s)\in {\cal R}}\hat{t}_{(S,s)}(F_u).$$
Let $u$ have the children $v_1,\ldots,v_d$ in $F$.

We have 
\begin{align*}
\hat{t}_{\emptyset}(F_u) & = \hat{t}(F_{v_1})\oplus \hat{t}(F_{v_2})\oplus \cdots \oplus \hat{t}(F_{v_d}).
\end{align*}
Now, let $(S,s)\in {\cal R}$. Let $s_1,\ldots,s_{d'}$ be the neighbors of $s$ in $S$
and let $S_i$ be the component of $S-s$ containing $s_i$.
Since $S$ has order at most $\Delta$, we have $d'<\Delta$.

Furthermore, we have 
\begin{align*}
\hat{t}_{(S,s)}(F_u) & = 
\bigcup\limits_{f:[d']\xrightarrow{injective}[d]}
\left(\bigoplus\limits_{i\in [d']}\hat{t}_{(S_i,s_i)}(F_{v_{f(i)}})\,\,
\oplus
\bigoplus\limits_{i\in [d]\setminus f([d'])}\hat{t}_{\emptyset}(F_{v_i})\right),
\end{align*}
where the $O(d^\Delta)$ injective functions $f$ capture the different ways of 
associating the neighbors of $s$ in $S$ with the children of $u$ in $F$.
See Figure \ref{fign2} for an illustration.

\begin{figure}[htb!]
    \centering
    \begin{tikzpicture} [scale=0.3]
    \tikzstyle{point}=[draw,circle,inner sep=0.cm, minimum size=1mm, fill=black]
    \tikzstyle{point2}=[draw,circle,inner sep=0.cm, minimum size=0.5mm, fill=black]

    \tikzstyle{line1}=[line width=0.5mm]

    \tikzstyle{point3}=[draw,circle,inner sep=0.cm, minimum size=1.5mm, fill=none]

\xdef\www{5} 
\xdef\hhh{13.9*0.5}
\xdef\w{\www} 
\xdef\www{\www*0.75/0.5} 

\coordinate (s) at (3.5*\www,13) [label=above:\text{$\pi(s)=u$}] {};
\node[point] at (3.5*\www,13) [label=above:\text{$\pi(s)=u$}] {};
    
\foreach \i in {1,2}{
    \xdef\hhh{13.9*0.5}
    \xdef\h{13.9*0.75}

    \begin{scope}[shift={(\www*\i,0)}]
    \ifthenelse{\i=1}{
        \coordinate (si) at (0,\hhh-1.75) {};
        \node[point] at (0,\hhh-1.75) [label=left:\text{$\pi(s_1)=v_{f(\i)}\text{ }$}] {};
    }{
        \coordinate (si) at (0,\hhh-1.75) {};
        \node[point] at (0,\hhh-1.75) [label=left:$v_{f(\i)}\text{ }$] {};
    }
        \draw[line1] (s) -- (si);
        \node[align=center] at (0,\hhh-1.75-2.5) {$\cong$};
        \node[align=center] at (0,\hhh-1.75-4) {$S_{\i}$};

    \draw [rounded corners=5mm] (-\www/2,\hhh-\h)--(0,\hhh)--(\www/2,\hhh-\h)--cycle;

    \draw [line1, rounded corners=5mm] (-\w/2,0)--(0,\hhh)--(\w/2,0)--cycle;

\node at (0,0.25*\h-4-4) [label=left:]{$F_{v_{f(\i)}}$};
 
 \end{scope}
}

\foreach \i in {3.5}{
    \xdef\hhh{13.9*0.5}
    \xdef\h{13.9*0.75}

    \begin{scope}[shift={(\www*\i,0)}]
        \coordinate (si) at (0,\hhh-1.75) {};
        \node[point] at (0,\hhh-1.75) [label=left:$v_{f(d')}\text{ }$] {};

        \node[align=center] at (0,\hhh-1.75-2.5) {$\cong$};
        \node[align=center] at (0,\hhh-1.75-4) {$S_{d'}$};
            
        \draw[line1] (s) -- (si);
    \draw [rounded corners, rounded corners=5mm] (-\www/2,\hhh-\h)--(0,\hhh)--(\www/2,\hhh-\h)--cycle;

    \draw [line1,rounded corners=5mm] (-\w/2,0)--(0,\hhh)--(\w/2,0)--cycle;
        \node at (0,0.25*\h-4-4) [label=left:]{$F_{v_{f(d')}}$};
    \end{scope}
}

\node[align=center] at (23.5,9) {$\overset{\bm{\pi}}{\bm{\cong} }\mathbf{S}$};

\foreach \i in {-1,0,1}{
    \node[point2] at (\i+2.75*\www-1.5,\hhh-1.75) {};
}

\foreach \i in {-1,0,1}{
        \node[point2] at (\i+5.25*\www,\hhh-1.75) {};
}

\foreach \i in {4.5,6}{
    \xdef\hhh{13.9*0.5}
    \xdef\h{13.9*0.75}
    \begin{scope}[shift={(\www*\i,0)}]
        \node[point3] (si) at (0,\hhh-1.75) [label=below:] {};
        \draw (s) -- (si);
        \draw [rounded corners, rounded corners=5mm] (-\www/2,\hhh-\h)--(0,\hhh)--(\www/2,\hhh-\h)--cycle;
        \node at (0,0.33*\h-4) [label=left:]{};
    \end{scope}
}

\draw[thick,black,decorate,decoration={brace,amplitude=10}] (6.5*\www,-4.5) -- (4*\www,-4.5) node[midway, below,yshift=-10]{$ \bigcup\limits_{i\in [d] \setminus f([d'])} F_{v_i}$};    

\end{tikzpicture}

\caption{Embedding $S$ (as well as the rest of $K$) into $F_u$ 
mapping the root $s$ of $S$ to $u$
and the children $s_i$ of $s$ in $S$ to children $v_{f(i)}$ of $u$
as selected by $f$.}\label{fign2}
\end{figure}
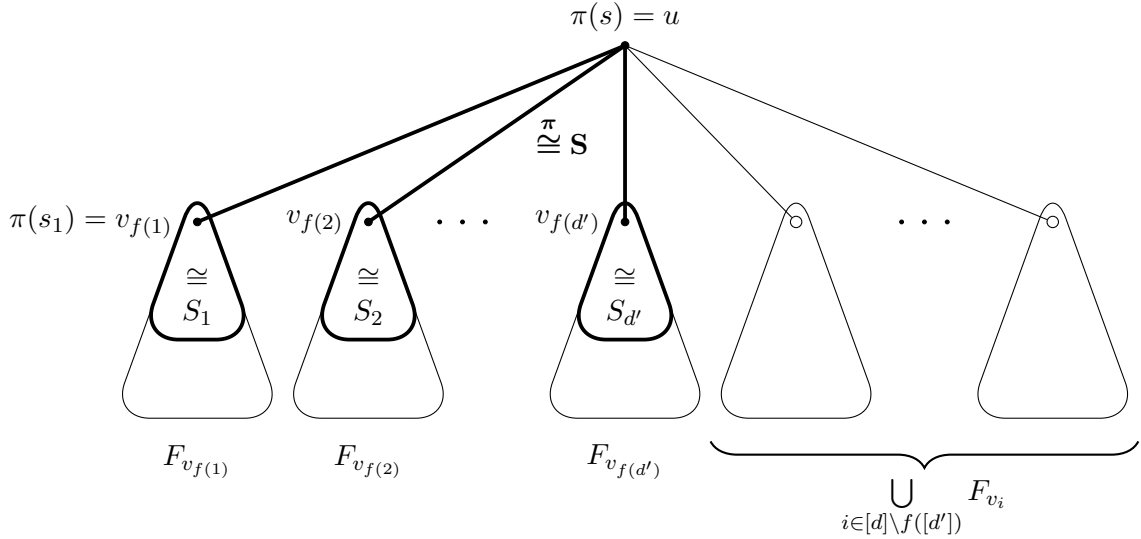
Using these formulas, a simple dynamic programming approach allows to determine
$\hat{t}(F)$ as well as suitable realizers within $F$ in time $O(n^p)$.
\end{proof}

\pagebreak

\begin{proof}[Proof of Theorem \ref{theorem4}]
Let $\epsilon>0$ be fixed.
Let ${\cal F}=\{ F_1,\ldots,F_k\}$ be a given set of $k$ forests of order at most $n$.
For $i\in [k]$, let $n_i=n(F_i)$, and let $n_1=\min\{ n_1,\ldots,n_k\}$.
Let $F_{\rm opt}$ be some maximum subforest of ${\cal F}$.
Since each forest $F_i$ in ${\cal F}$ has an independent set of order at least $n_i/2\geq n_1/2$,
we have $n(F_{\rm opt})\geq n_1/2$.

Let $\Delta=\left\lceil\frac{2}{\epsilon}\right\rceil$.
Let ${\cal F}_{\Delta}$ be as above, that is, 
${\cal F}_{\Delta}$ is the set of all forests 
whose components all have order at most $\Delta$.
Rooting each component of $F_1$ in some vertex and 
iteratively removing vertices $u$ of maximum depth 
for which $u$ has at least $\Delta$ descendants, 
yields a set $X$ of at most $n_1/\Delta\leq 2n(F_{\rm opt})/\Delta$ vertices of $F_1$ 
such that $F'_1=F_1-X$ belongs to ${\cal F}_{\Delta}$.
Let $F'_{\rm opt}$ be a maximum subforest of $({\cal F}\setminus \{ F_1\})\cup \{ F'_1\}$.
Clearly, $F'_{\rm opt}$ is a subforest of ${\cal F}$ that belongs to ${\cal F}_{\Delta}$ and satisfies
\begin{align*}
n(F'_{\rm opt})
\geq n(F_{\rm opt})-|X|
\geq \left(1-\frac{2}{\Delta}\right)n(F_{\rm opt})
\geq (1-\epsilon)n(F_{\rm opt}).
\end{align*}
Therefore, in order to complete the proof, 
it suffices to show that a subforest of ${\cal F}$ that belongs to ${\cal F}_{\Delta}$
and has maximum possible order subject to this condition,
can be found efficiently.

By Lemma \ref{lemma1}, we can determine $\hat{t}(F_i)$ 
as well as suitable realizers within $F$ 
in time $O(n^{p_1})$ for every $i\in [k]$.
Since 
\begin{align*}
\max\{ n(F):F\in{\cal F}_{\Delta}\mbox{ is a subforest of }{\cal F}\}
&=\max\left\{
\sum\limits_{i=1}^q t_in(T_i):
(t_1,\ldots,t_q)\in \bigcap\limits_{i=1}^k\hat{t}(F_i)\right\},
\end{align*}
the desired statement follows.
\end{proof}
It seems interesting to study tradeoffs between supergraphs 
that are required to belong to different graph classes. 
For a set ${\cal F}$ of trees, for instance, 
a supergraph of minimum order may be much smaller than a minimum supertree. 
Indeed, 
if ${\cal F}=\{ T(a,0,a),T(a,0,0,a),\ldots,T(a,\underbrace{0,\ldots,0}_{k\,\,times},a)\}$
for positive integers $a$ and $k$ at least $3$,
then suitably identifying vertices of degree $a+1$ yields a supergraph of ${\cal F}$
of order $2+2a+1+2+\ldots+k=2a+{k+1\choose 2}+2$,
while every supertree of ${\cal F}$ has order $\Omega\left(\sqrt{k}a\right)$.

\end{document}